\documentclass[preprint,authoryear,12pt]{elsarticle}

\usepackage[authoryear]{natbib}
\usepackage{amssymb,latexsym}
\usepackage{graphicx,amsmath,amsfonts,amsthm}

\def\R{\mathbb{R}}

\def\row#1#2{{#1}_1,\ldots ,{#1}_{#2}}

\def\2vec#1#2{\left(\begin{array}{c}{#1}\\{#2}\end{array}\right)}
\newtheorem{theorem}{Theorem}
\newtheorem{corollary}{Corollary}
\newtheorem{lemma}{Lemma}

\newtheorem{proposition}{Proposition}
\newtheorem{example}{Example}

\newtheorem{definition}{Definition}

\journal{Mathematical Social Sciences}

\begin{document}

\begin{frontmatter}
\title{\bf Growth of Dimension in Complete Simple Games}
\author{\bf Liam O'Dwyer and Arkadii Slinko}
%\date{\today}
\address{The University of Auckland, New Zealand}
%\begin{document}

%\maketitle

\begin{abstract}
The concept of dimension in simple games was introduced by  \cite{taylor93} as a measure of the remoteness of a given game from a weighted game. They demonstrated that the dimension of a simple game can grow exponentially in the number of players. However,  the problem of worst-case growth of the dimension in complete games was left open. \cite{freixas08} showed that complete games of arbitrary dimension exist and, in particular, their examples demonstrate that the worst-case growth of dimension in complete games is at least linear. In this paper,  using a novel technique of  \cite{Kurz2015}, we demonstrate that the growth of dimension in complete games can also be exponential in the number of players. 
\end{abstract}
\begin{keyword}

simple game \sep complete simple game \sep dimension \sep boolean dimension



\end{keyword}

\end{frontmatter}

\section{Introduction}
The past two decades have witnessed an explosion of interest in computational and representational issues related to coalitional games (see, e.g., \cite{Deng:1994,Ieong:2005}) and simple games, in particular (see, e.g., \cite{deineko06,elkind08,faliszewski09}).
Simple games were introduced in their present form by  \cite{vonneumann44} for applications in economics but found a wide range of applications across several disciplines. In particular, simple games are  used to model  decision making in committees \citep{peleg2002game}, reliability of real life  systems made from unreliable components \citep{ramamurthy2012coherent} and McCulloch-Pitts units in threshold logic \citep{muroga71}. 

A simple game consists of a finite set of players and a set of winning coalitions that satisfies the monotone property asserting that all supersets of a winning coalition are also winning.
One of the most important classes of simple games is the class of \emph{weighted simple games}. In a weighted simple game every player is assigned a non-negative real  weight  so that a coalition is winning if the total weight of its players is at least some predetermined threshold.  
From the computational perspective weighted games are especially important since they can be succinctly represented by a finite sequence of integers \citep{freixas97}. If we allow the weights and the threshold to be vector-valued, then every simple game becomes weighted \citep{taylor93} and the smallest dimension of vectors which makes this representation possible is called the \textit{dimension} of the game. 

\cite{taylor99} demonstrated that the dimension of simple games can grow exponentially in the number of players. Recently \cite{Olsen2016} established the exact nature of this growth which is $2^{n-o(n)}$, where $n$ is the number of players.
 
Another important class of simple games is the class of complete games introduces by \cite{carreras96}. In a complete game it is always possible to say which player among any two players is at least as desirable (as a coalition partner) as another one and, moreover, this desirability relation is a total order. This is a much broader class of games than weighted games which includes, for example, conjunctive and disjunctive hierarchical games which appear as the access structures of popular secret sharing schemes \citep{simmons90,tassa07}. Both disjunctive and conjunctive hierarchical games are seldom weighted \citep{gvozdeva13} or even roughly weighted \citep{Hameed2015}. 

  \cite{freixas08} studied conjunctive hierarchical games (under the name of games with a minimum)  and found that their dimension  grows linearly in the number of players and asked  whether or not in the class of complete games the dimension can grow polynomially or even exponentially. In the present paper we show that the growth of dimension of disjunctive hierarchical games, which are complete, is exponential in the number of players. This, in particular, fully answers the question of Freixas and Puente. 
  
To obtain the aforementioned result we had to find a class of complete games whose dimension grows very fast. We found this growth in the class of disjunctive hierarchical games. We also turned our attention to conjunctive games for two reasons. Firstly, we fixed the gap in the aforementioned result of Freixas and Puente and also removed an unnesessary requirement of absence of veto and dummy players. The gap in their theorem is non-trivial and goes to the heart of the definition of dimension. In a way, this definition is less nice than Freixas and Puente thought. We devoted Section~4 for clarification of this problem.

\section{Preliminaries}
\smallskip
\noindent\textbf{Simple Games.} 
\medskip

Let $P$ be a set consisting of $n$ players.  For convenience $P$ can be taken to be $[n]=\{1,2,\ldots,n\}$.
\begin{definition}
A simple game is a pair $G=(P,W)$, where $W$ is a subset of the power set $2^P$  which satisfies the monotonicity condition: 
\begin{quote}
if $X\in W$ and $X\subset Y\subseteq P$, then $Y\in W$.
\end{quote}
%We also require that $W$ is different from $\emptyset$ and $P$ (non-triviality assumption).
\end{definition}

Elements of the set $W$ are called {\em winning coalitions}. We also define the set $L=2^P\setminus W$ and call elements of this set {\em losing coalitions}. A winning coalition is said to be {\em minimal} if every proper subset of it is a losing coalition. A losing coalition is said to be {\em maximal} if every proper superset of it is winning. Due to monotonicity, every simple game is fully determined by the set of its minimal winning coalitions $W_{\text{min}}$ or the set of maximal losing coalitions $L_{\text{max}}$.\par\medskip

\noindent\textbf{Weighted Simple Games and Criteria of Weightedness.}
\begin{definition}
A simple game $G$ is called a {\em weighted (majority) game} if there exist non-negative reals $\row wn$, and a positive real number $q$, called the quota, such that $X\in W$ iff $\sum_{i\in X} w_i\ge q$. Such a game is denoted $[q;\row wn]$. We also call $[q;\row wn]$ a {\em voting representation} for $G$.
\end{definition}

\begin{example}
Let now $n=2k-1$ be odd and  $W$ be all subsets of $P$ of cardinality  $k$ or greater. There are exactly $2^{n-1}$ elements in $W$. This game is called the simple majority voting game. It is weighted and $[k;1,1,\ldots,1]$ is its voting representation.
\end{example}

A more interesting example is 
\begin{example}%[\cite{FM}]
The UN Security Council consists of five permanent and 10 non-permanent members (which are sovereign states). A passage requires approval of at least nine countries, subject to a veto by any one of the permanent members. This is a weighted simple game with a voting representation
\[
[39;7,7,7,7,7,1,1,1,1,1,1,1,1,1,1].
\]
\end{example}

A sequence of coalitions 
\begin{equation*}
\label{tradingtransform}
{\cal T}=(\row Xj;\row Yj)
\end{equation*}
of simple game $G$ is a \textit{ trading transform} of length $j$ if the coalitions $\row Xj$ can be converted into the coalitions $\row Yj$ by rearranging players. It can also be expressed as 
\[
|\{i:a\in X_i\}| = |\{i:a\in Y_i\}|\qquad \text{for all $a\in P$}.
\]
A trading transform $\mathcal{T}$ is called a \textit{certificate of non-weighted\-ness} for $G$ if $\row Xj$ are winning in $G$ and $\row Yj$ are losing. The absence of certificates of non-weightedness of any length is a necessary and sufficient condition of weightedness  of the game $G$ \citep{Elgot,taylor99}. \par\medskip

A more general class of games that we will touch upon is the class of roughly weighted games \citep{gvozdeva13}.

\begin{definition}%[\citeA{GS2011}]
\label{roughlyweighted}
A simple game $G$ is called {\em roughly weighted} if there exist non-negative real numbers $\row wn$ and a  real number $q$, called the {\em quota}, not all equal to zero, such that for a coalition $X\in 2^P$ the condition $\sum_{i\in X} w_i< q$ implies $X$ is losing, and $\sum_{i\in X} w_i> q$ implies $X$ is winning.  %We say that $[q;\row wn]$ is  a {\em rough voting representation} for~$G$.
\end{definition}

Weighted games and a roughly weighted games both have a system of weights and a threshold. The difference is the treatment of coalitions on the threshold. In the weighted case all of them are winning and in the roughly weighted case both winning and losing coalitions can occur.\par\medskip

Finally a few words about %subgames, reduced games and 
duality. 

%\begin{definition}
%Let $G=(P,W)$ be a simple game, $A$ be a proper subset of $P$ and $A^c$ is the complement of $A$ in $P$. Let us define subsets %$W_{\text{sg}}\subseteq W$ %and $W_{\text{rg}}\subseteq W$ by
%\[
%W_{\text{sg}}=\{X\subseteq  A^c\mid X\in W\}, \quad
%W_{\text{rg}}=\{X\subseteq A^c\mid X\cup A\in W\}.
%\]
%Then the game $G_A=(A^c,W_\text{sg})$ is called a {\em subgame} of $G$ and $G^A=(A^c,W_\text{rg})$ is called a {\em reduced game} of $G$.
%Both subgames and reduced games for brevity will be called {\em minors}.
%\end{definition}

%It is easy to show that every subgame and every reduced game of a weighted  game is also a weighted  game. %For example, in the case of subgame one just has to retain the same weights for elements of $A^c$ as in $G$ and the same threshold.

Let $G=(P,W)$ be a simple game and $L$ be the set of its losing coalitions. We define the game $G^\ast=(P,W^\ast)$ dual to $G$ by setting
\[
W^\ast=\{P\setminus X \mid X\in L\},
\]
i.e., the winning coalitions of $G^\ast$ are complements to the losing coalitions of~$G$.
\par\medskip

\noindent\textbf{Complete and Hierarchical Simple Games.}
\medskip

Given a simple game $G=(P,W)$, %on the set of players $P$ 
after \cite{isbell58}, we define a relation $\succeq_G$ on $P$ by setting $i \succeq_G j$ if for every set $X\subseteq P$ not containing $i$ and~$j$ 
\begin{equation}
\label{condition}
X\cup \{j\}\in W \Longrightarrow X\cup \{i\} \in W.
\end{equation}
In such a case we will say that $i$ is at least as {\em desirable} (as a coalition partner) as $j$. % In the United Nations Security Council every permanent member will be more desirable than any non-permanent one.
This relation is reflexive and transitive but not always complete (total) (e.g., see \cite{carreras96}). The corresponding equivalence relation on $[n]$ will be denoted $\sim_{G} $ and the strict desirability relation as $\succ_G$. If this can cause no confusion we will omit the subscript $G$. 

\begin{definition}
A game whose desirability relation is complete is called {\em complete}. 
\end{definition}

\begin{example}
Any weighted game is complete.% \par\medskip
\end{example}
Later we will have more examples. Complete simple games are a very natural generalisation of weighted games. This class is much larger, however, so measures of non-weightedness, e.g., the dimension, for such games are important and interesting. \par

In a complete game $G=(P,W)$ the set of players $P$ is partitioned into equivalence classes $P=P_1\cup\ldots\cup P_m$ with respect to $\sim_{G} $. Without loss of generality we will consider that
\begin{equation}
\label{mclasses}
P_1\succ_G P_2\succ_G \ldots \succ_G P_m.
\end{equation}
Such game $G$ is called $m$-{\em partite}.

Any coalition $X\subseteq P$ defines a multiset $\{1^{\ell_1},\ldots, m^{\ell_m}\}$, where $\ell_i$ is the number of elements from $P_i$ in $X$. Due to completeness, the status of a coalition $X$, i.e., whether it is winning or losing, can be deduced from this multiset. The multisets corresponding to winning coalitions will be called {\it models of winning coalitions}. We can also define {\it models of losing coalitions}, respectively.

In a complete game $G=(P,W)$ a winning coalition $X$ is {\em shift-minimal}  if any coalition $(X\setminus \{i\})\cup \{j\}$ is losing for any $i\in X$ and $j\notin X$ such that $i\succ_Gj$, i.e., it ceases to be winning after any replacement of its player with a less desirable one. A losing coalition $Y$ is called {\em shift-maximal} if it becomes winning after a replacement of any player with a more desirable player. A complete simple game is fully defined by the set of its shift-minimal winning coalitions or a set of its shift-maximal losing coalitions.

Suppose now that the set of players $P$ is partitioned into $m$ disjoint subsets $P=\cup_{i=1}^m P_i$  and let $k_1<k_2<\ldots<k_m$ be a sequence of positive integers. Let ${\bf k}=(\row km)$. Then we define the game $H=H_{\exists}(P,{\bf k})$ by setting  the set of winning coalitions to be
\[
W_\exists = \left\{ X\in 2^P\mid \exists i \left(\left|X\cap \left(\cup_{j=1}^i P_i\right)\right|\ge k_i\right) \right\}.
\]
Such a game is called a \textit{disjunctive hierarchical game}. It has $m$ thresholds and one of them must be reached for the coalition to be winning.

Suppose now that the set of players $P$ is partitioned into $m$ disjoint subsets $P=\cup_{i=1}^m P_i$, and let $k_1< \ldots < k_{m-1}\le k_m$ be a sequence of positive integers. Then we define the game $H_{\forall}(P,{\bf k})$ by setting the set of its winning coalitions to be
\[
W_\forall = \left\{ X\in 2^P\mid \forall i \left(\left|X\cap \left(\cup_{j=1}^i P_i\right)\right|\ge k_i\right) \right\}.
\]
%\end{definition}
Such a game is called a \textit{conjunctive hierarchical game}. It has $m$ thresholds and all of them must be reached for the coalition to be winning.

Both classes of hierarchical games are complete.  \cite{gvozdeva13} give a sufficient and necessary conditions for the game $H_{\forall}(P,{\bf k})$ defined above to be truly  $m$-partite\footnote{For some combinations of parameters the number of equivalence classes may in fact be less than $m$.}. We denote $|P_i|=n_i$, then the following two conditions jointly are necessary and sufficient:
\begin{eqnarray}
\label{k_1le n_1}
k_1&\le& n_1, \\
\label{sizeofk_i}
k_i&<&k_{i-1}+n_i 
\end{eqnarray}
for every $i\in \{2,\ldots, m\}$.
Moreover,  \cite{gvozdeva13} showed that $G$ has veto players if and only if $k_1=n_1$, in which case $P_1$ is the set of veto players, and $G$ has dummy players if and only if $k_{m-1}=k_m$, in which case $P_m$ is the set of dummy players.

We note that these conditions imply 
\begin{equation}
\label{sumofns}
k_i\le n_1+\ldots +n_i-(i-1).
\end{equation}
for all $i=1,\ldots,m$, and 
\begin{equation}
\label{n_i>1}
n_i>1.
\end{equation}
for all $1<i<m$, moreover, it is also true for $i=m$ in the absence of dummy players. 
The first inequality follows from \eqref{sizeofk_i}. 
In particular, it shows that the equation $k_i =  n_1+\ldots +n_i$ can be satisfied only when $i=1$.
The second also follows from \eqref{sizeofk_i} since 
\[
n_i>k_i-k_{i-1}\ge 1
\]
(for $i=m$ this needs $k_m\ne k_{m-1}$, which means no dummies.

\begin{lemma}
\label{reduced}
Let $G=(P,W)$ be an $m$-partite simple game with $P=P_1\cup P_2\cup \ldots \cup P_m$, where $P_1$ consists of veto players and $P_m$ of dummy players. Let $A=P_1\cup P_m$. Then the reduced game $G^A$ is defined on the set of players $P_2\cup \ldots\cup P_{m-1}$ and does not have veto or dummy players).
\end{lemma}

\begin{proof}
Due to Proposition~5 of \cite{gvozdeva13} $G_A$ is an $(m-2)$-partite conjunctive hierarchical game with ${\bf n}'=(n_2,\ldots,n_{m-1})$ and with the vector of thresholds
\[
{\bf k}'=(k_2-k_1,\ldots, k_{m-1}-k_1).
\]
Since $n_2>k_2-k_1$ this game does not have veto players and since $k_{m-2}-k_1<k_{m-1}-k_1$ it does not have dummies.\vspace{2mm}
\end{proof}

\noindent\textbf{Definition of the Dimension. The Criterion of Kurz and Napel.}
%\subsection{Definition of the Dimension. The Criterion of Kurz and Napel}
\medskip

%\begin{definition} 
The \emph{dimension} of a simple game $G$ is the minimum dimension of the vectors required to express it as a vector-weighted game \citep{taylor93}.  That is, a simple game has dimension $k$ if it can be represented as a vector-weighted game with weights from $\R^k$, but not with weights from $\R^{\ell}$ for $\ell<k$. %We will denote the dimension of $G$ by $\dim(G)$.
%\end{definition}
%
In practice it is more convenient to work with the following equivalent definition.\par\smallskip
%
%\begin{definition}
Let $G_i=(P, W_i)$,  $i=1,\ldots,n$, be simple games on the same set of players~$P$. Then the \emph{intersection} of these games is the simple game $G = (P, W)$, where $W = W_1 \cap W_2 \cap \cdots \cap W_n$. In other words, a coalition is winning in $G$ if and only if it is winning in $G_i$ for each $i= 1, \ldots, n$. We write $G=G_1\wedge \ldots \wedge G_n$ for reasons that will be revealed later.\par\smallskip
%\end{definition}

\begin{definition}
A simple game has dimension $d$ if it can be represented as the intersection of $d$ weighted games but cannot be represented as the intersection of $\ell$ weighted games for $\ell<d$. We will denote the dimension of $G$ by $\dim(G)$.
\end{definition}

%Taylor and Zwicker  \cite{taylor99} showed that a simple game $G = (P,W)$ has dimension $n$ if and only if it can be represented as the intersection of $n$ weighted games $(P, W_j)$, but no fewer.

Examples of games of dimension 2 include the United States Federal System and the procedure to amend the Canadian Constitution \citep{taylor99}. \cite{freixas04} showed that the dimension of the European Union Council under the Nice rules had dimension 3. In a recent article \cite{Kurz2015} have found that the revised voting rules of the Council of the European Union (EU Council) mean that a simple game representation of that voting body must have dimension at least 7. This is significantly larger than that of any other known simple game that occurs in the real world.

To calculate the dimension of a game exactly is not an easy task.  \cite{deineko06} proved that the following problem  is NP-hard: given $k$ weighted majority games on the same set of players, decide whether the dimension of their intersection is exactly~$k$.  

To bound the dimension of a game from above the following observation can be used \citep{taylor99}.

\begin{proposition}
\label{Lmax}
The dimension of a simple game $G=(P,W)$ is at most the cardinality $|L_{max}|$ of the set  $L_{max}$ of maximal losing coalitions of~$G$.
\end{proposition}

To bound the dimension of a game from below we will use the following useful criterion, which is Observation~1 in  \cite{Kurz2015}. It is so important that we call it a theorem.

\begin{theorem}
 \label{kurz}
Let $G = (P,W)$ be a simple game, and let $S=\{\row Yk\}$ be a set of losing coalitions such that for each pair $\{Y_i, Y_j\}$ with $i \neq j$, there is no weighted simple game for which every coalition in $W$ is winning but $Y_i$ and $Y_j$ are both losing. Then the dimension of $G$ is at least $|S|=k$.
\end{theorem}

Kurz and Napel refer to these elements of $S$ as \emph{pairwise incompatible}. One way to use this theorem to prove that a simple game $G$ has dimension at least $k$ is to find, for every pair  $\{Y_i, Y_j\}\subseteq S$, a certificate of non-weightedness  $(X^1_{i,j}, X^2_{i,j}; Y_i,Y_j)$, where $X^1_{i,j},$ and $X^2_{i,j}$ are both winning in~$G$.

\section{The Main Results}
\smallskip

\noindent\textbf{Dimension of Disjunctive Hierarchical Games.} 
\medskip

Firstly, let us consider a non-weighted example of a disjunctive hierarchical game with the `smallest' possible vector ${\bf k}$. This would be ${\bf k}=(2,4)$ (as in non-trivial cases we have $k_1\ge 2$ and the games with ${\bf k}=(2,3)$ are weighted \citep{gvozdeva13}). Although not weighted, it is known \citep{gvozdeva13}, that the game with this set of parameters is always roughly weighted.

\begin{proposition}
Let $d \geq 2$ be a positive integer.
Let $P=P_0\cup P_1$ with $|P_0|=d$, $|P_1|=2d$, and ${\bf k}=(2,4)$. 
Then the disjunctive hierarchical game $H=H_\exists(P,{\bf k})$  has dimension at least $d$.
\end{proposition}

\begin{proof}Let $P_0 = \{a_0, \ldots, a_{d - 1}\}$ and $P_1 = \{b_0, \ldots, b_{2d - 1}\}$, define the sets  $Y_i = \{a_i, b_{2i}, b_{2i + 1}\}$, $i = 0, \ldots,  d - 1$, and let $S= \{Y_0, \ldots, Y_{d - 1}\}$. All coalitions from $S$ lose in $H$ since they have neither two players from $P_1$, nor four players in total. Then $S$ satisfies the conditions of the Kurz-Napel criterion since if $i \neq j$ 
\[
( \{a_i, a_j\}, \{b_{2i}, b_{2i + 1}, b_{2j}, b_{2j + 1}\}; Y_i, Y_j )
\]
is a certificate of non-weightedness for $H$ as the coalitions $X^1_{i,j}=\{a_i,a_j\}$ and $X^2_{i,j} = \{b_{2i}, b_{2i + 1}, b_{2j}, b_{2j + 1}\}$ are both winning (the first achieves the first threshold and the second achieves the second). By the criterion, the dimension of $H$ is at least $d$.
\end{proof}

%We have shown here that we have dimension at least $|P_0|$ when $|P_1| = 2|P_0|$. If $|P_1| > 2|P_0|$, then we have the same result, as we can use exactly the same set of pairwise incompatible losing coalitions. In this case our coalitions will not include all the members of $P_1$, however. Finally, if $|P_1| < 2|P_0|$, then we can again use the same set of pairwise incompatible coalitions, but in this case we will not use every member of $P_0$. We will end up with $\left \lfloor{\frac{|P_1|}{2}} \right \rfloor$ pairwise incompatible losing coalitions in this case, since if $|P_1|$ is odd we will have only $\frac{|P_1| - 1}{2}$ pairs of members from $|P_1|$. We must have that $|P_1| \geqslant 2$ so that we can form the pair $\{b_{2i}, b_{2i + 1}\}$. 
Thus we have the following result.

\begin{theorem}
There exist roughly weighted games of arbitrary large dimension.
\end{theorem}

Interestingly, we showed that we can get linear growth in the number of players without increasing the number of classes of equivalent players. If we start increasing both we will get a growth faster than linear. 

\begin{lemma} 
\label{OS3}
Let $P=P_0\cup P_1\cup \ldots\cup P_{m-1}$ with $|P_0| = k$, $|P_1| = |P_2| = \cdots = |P_{m - 1}| = 2k$ and ${\bf k}=(2,4,6, \ldots, 2m)$. 
Then the disjunctive hierarchical simple game $H=H_\exists (P,{\bf k})$  has dimension $d$ satisfying % at least $k^{m - 1}$.
\begin{equation}
\label{inford}
k^{m - 1}\le d\le k^m (2k - 1)^{m - 1}. 
\end{equation}
\end{lemma}

\begin{proof}
Let $P_0=\{\row ak\}$ and denote by $p^{(j)}_i$, $j \in \{0, \ldots, 2k - 1\}$, the $j$th player from part $P_i$. We will also denote $P^{(j)}_i=\{p^{(2j)}_{i}, p^{(2j + 1)}_{i}\}$, where $j \in \{0, \ldots, k - 1\}$ and $i\in [m]$. Then all coalitions of the form 
\begin{equation}
\label{coalition_in_S}
\{a_{i_0}\}\cup  P^{(i_1)}_{1} \cup P^{(i_2)}_{2}\cup \cdots \cup P^{(i_{m - 2})}_{m - 2}\cup P^{(i')}_{m - 1},
\end{equation}
 where $i'=i_0+ i_1+ \ldots + i_{m - 2} \pmod k$ will form the set $S$ to be used in the Kurz-Napel criterion. Firstly, we note that all the coalitions in $S$ are losing as no threshold is achieved. Let us show that any two of them are incompatible. Let
 \begin{align*}
 Y_1&=\{a_{i_0}\}\cup  P^{(i_1)}_{1} \cup P^{(i_2)}_{2}\cup \cdots \cup P^{(i_{m - 2})}_{m - 2}\cup P^{(i')}_{m - 1},\\
 Y_2&=\{a_{j_0}\}\cup  P^{(j_1)}_{1} \cup P^{(j_2)}_{2}\cup \cdots \cup P^{(j_{m - 2})}_{m - 2}\cup P^{(j')}_{m - 1},
 \end{align*}
be two coalitions from $S$. There are two cases.

\begin{itemize}
\item If $i_0 = j_0$, then there is at least one $\ell  \in \{1, \ldots, m - 2\}$ for which $i_\ell  \neq j_\ell $ (otherwise the coalitions would be identical), so $P^{(i_\ell )}_{\ell }$ is disjoint from $P^{(j_\ell )}_{\ell }$. If $i_r = j_r$ for all $r \in \{1, \ldots, \ell  - 1, \ell  + 1, \ldots, m - 2\}$, then $i' \neq  j'$ , thus we may assume that there exists also $r\in [m-1]$ such that $r\ne \ell$ and $i_r \ne j_r$. Without loss of generality assume that $\ell<r$. Then we get a certificate of non-weightedness
\[
(Y_1\cup \{p^{(2j_\ell )}_{\ell}\}\setminus P^{(i_r)}_{r}, (Y_2\setminus \{p^{(2j_\ell )}_{\ell}\}\cup P^{(i_r)}_{r}    ; Y_1,Y_2).
\]
by swapping one element $\{p^{(2j_\ell )}_{\ell }\}$ of $Y_2$ for two elements of $P^{(i_r)}_{r}$ from $Y_1$.  After the swap the first coalition will be winning since the $\ell$-th threshold is achieved and the second will be also winning since the $r$-th threshold is achieved.

\item If $i_0 \neq j_0$ but $i_r = j_r$ for all $r \in \{1, \ldots, m - 2\}$, then $i' \neq j'$, hence we may assume that there exists $r\in [m-1]$ such that $i_r\ne j_r$. In this case we get a certificate of non-weightedness
\[
(Y_1\cup \{a_{j_0} \}\setminus P^{(i_r)}_{r}, (Y_2\setminus \{a_{j_0} \}\cup P^{(i_r)}_{r}    ; Y_1,Y_2)
\]
by swapping $\{a_{j_0}\}$ and $P^{(i_r)}_{r}$. 
\end{itemize}
Since $|S|=k^{m - 1}$, by Theorem~\ref{kurz}, this means that the dimension of such a game is at least $k^{m - 1}$.

The upper bound is easily calculated with the help of Proposition~\ref{Lmax} taking in consideration that each maximal losing coalition consists of one member from $P_0$ and two members from each of  the $P_1, \ldots, P_n$.
\end{proof}

\begin{theorem}
In the class of disjunctive hierarchical games with bounded number of equivalence classes the worst-case growth of dimension is polynomial.  
\end{theorem}

\begin{proof}
If in \eqref{inford} we fix $m$, then we have $(2k+1)m$ voters which is linear in $k$ and the dimension which is polynomial in $k$. Hence the dimension growth is polynomial of degree $m$.
\end{proof}

\begin{theorem}
\label{OS2}
In the class of disjunctive hierarchical games with bounded number of players in equivalence classes the worst-case growth of dimension is exponential.  
\end{theorem}

\begin{proof}
If in \eqref{inford} we fix $k$, then we have $(2k+1)m$ voters which is linear in $m$ and the dimension is exponential in $m$. Hence the dimension growth is exponential in~$m$.
\end{proof}

These results answer directly the question from \cite{freixas08} about possibility of a polynomial or exponential growth in complete simple games.\par\bigskip

\noindent\textbf{Dimension of Conjunctive Hierarchical Games.} 
\medskip

 \cite{freixas08} studied a class of games, that they called games with minimum, which, as was proved in \cite{gvozdeva13}, is nothing other than the class of conjunctive hierarchical games. The %formulated the following theorem. 
theorem they formulated state that the dimension $d$ of an $m$-partite conjunctive  hierarchical game without veto or dummy players satisfies the inequalities ${\lceil \frac{m}{2} \rceil \le d \le m}$. 
%
%\begin{theorem}[Freixas-Puente, 2008]
%\label{FP}
%The dimension $d$ of an $m$-partite conjunctive  hierarchical game without veto or dummy players
%satisfies the inequalities ${\lceil \frac{m}{2} \rceil \le d \le m}$. %has dimension not exceeding~$m$.
%\end{theorem}
%
However, we have to reprove the lower bound in this theorem due to two reasons: 1) the gap in their proof which will be explained in the next section after rectification of the concept of dimension; 
2) the unnecessary requirement of having no dummies or vetoers. \par\medskip

But, firstly, we will give a simple  proof of the upper bound as well getting rid of unnecessary requirement of having no dummies or vetoers. 

\begin{proposition}
\label{upperbound}
The dimension $d$ of an $m$-partite conjunctive  hierarchical game is at most $m$.
\end{proposition}

\begin{proof}
Let $H=H_\forall (P,{\bf k})$, where $P=P_1\cup \ldots \cup P_m$ with ${\bf k}=(\row km)$, and suppose the game is $m$-partite. Let us denote $n=|P|$. Let us define $m$ weight functions on $P$ by
\[
    w_s(p)=
    \left\{
    \begin{array}{ll}
      1,&  \text{if}\ p\in P_1\cup \ldots \cup P_s \\
      0,&  \text{if}\ p\in P_{s+1}\cup \ldots \cup P_m
    \end{array}
    \right.
  \]
  and $m$ thresholds $q_s=k_s$, $s\in [m]$. We define game $G_s$ on $P$ by the weight function $w_s$ and threshold $q_s$. It is clear that a coalition $X$ wins in $G_s$ if and only if $ \left|X\cap \left(\cup_{j=1}^s P_i\right)\right|\ge k_s$, hence $H=G_1\wedge\ldots\wedge G_m$ and  $d\le m$.
\end{proof}

We will need the following lemma. 

\begin{lemma}
\label{shiftmaxl}
Let $H=H_\forall(P,{\bf k})$ be $m$-partite conjunctive hierarchical game with $P=P_1\cup\cdots\cup P_m$ being the equivalence classes of $\sim_H$, $|P_i|=n_i$ and ${\bf k}=(\row km)$. Then, if no dummies present, any shift-maximal losing coalition of $H$ corresponds to one of the following $m$ models:
\[
{\mathcal M}_i=\{1^{a_1},\ldots,i^{a_i}, (i+1)^{n_{i+1}},\ldots, m^{n_m}\}\quad (i=1,\ldots,m),
\]
where $\row ai$ satisfy the following conditions:
\begin{itemize}%\addtolength{\itemsep}{-5pt}
\item[(i)]  $a_1+\ldots+a_i=k_i-1$; 
\item[(ii)] If $0< a_t $ for $t\le i$, then $a_s=n_s$ for all $s<t$;
\item[(iii)]  $a_i< n_i$.
\end{itemize}
If dummies present, then the model ${\mathcal M}_m$ does not exist and we have only $m-1$ models.
\end{lemma}

\begin{proof}
Suppose the last $m$th level does not consist of dummies. Suppose that a maximal losing coalition $X$ fails the $i$th condition, i.e., in the first $i$ levels there are no $k_i$ elements, that is, $a_1+\ldots+a_i<k_i$. Due to \eqref{sumofns} we may assume that  $a_1+\ldots+a_i=k_i-1$. Since the $i$th threshold is already violated, we can include all elements of $P_{i+1}\cup\cdots\cup P_m$ without making this coalition winning. If for $i\le t$ we have $a_t>0$ but $a_{t-1}<n_{t-1}$, then we can replace one element of level $t$ with element of level $t-1$ without making coalition winning. As $X$ was shift-maximal, this is not possible. If the $m$th level consists of dummies, then $k_{m-1}=k_m$ so it is impossible to fail just one last threshold. Finally, if $a_i=n_i$, then by (ii) we have $a_s=n_s$ for all $s=1,\ldots, i-1$, in which case $X=P$.
\end{proof}

We illustrate this lemma with an example.

\begin{example}
Let us consider the game $H_{\forall}(P,{\bf k})$, where $|P_1|=|P_2|=|P_3|=4$, and ${\bf k}=(2,4,7)$. Then there are three models of shift-maximal losing coalitions:
\[
\{1,2^4,3^4\},\quad \{1^3,3^4\}, \quad \{1^4,2^2\}.
\]
The first type of coalitions fail the first threshold, the second type of coalitions fail the second threshold and the third type of coalitions fail the third threshold.
\end{example}

We can now prove the main result of this subsection.

\begin{theorem}
\label{OS}
The dimension $d$ of an $m$-partite conjunctive  hierarchical game  $H$
satisfies  ${\lceil \frac{m-1}{2} \rceil \le d \le m}$. If  $H$ has no dummy players, then ${\lceil \frac{m}{2} \rceil \le d \le m}$.
\end{theorem}

\begin{proof}%[Proof of Theorem~\ref{FP}]
It is easy to see that adding a level of dummy players does not change the dimension of the game so we may assume that $H$ has no dummy players. 
To prove the lower bound we assume that there are $\ell<\lceil \frac{m}{2} \rceil$ weighted games $\row H{\ell}$ such that $H=H_1\wedge \ldots \wedge H_\ell$. Then, in each game $H_i$, all winning coalitions of $H$ are winning and if a coalition is losing, then it is losing in one of the $H_i$. 

%Let $|P_i|=n_i$, $i=1,\ldots,m$.  By \eqref{n_i>1} we have $n_i\ge 2$ for all $i$ satisfying $1<i< m$. 

So there are $\ell<m/2$ games $\row H{\ell}$ but by Lemma~\ref{shiftmaxl} we have $m$ shift-maximal losing coalitions. Hence,
due to the pigeonhole principle there exist an index $i$ such that in $H_i$ at least three shift-maximal losing coalitions of $H$, belonging to different models, are losing in~$H_i$. Then there are two coalitions among these, say $L_1$ and $L_2$, whose corresponding models, say ${\mathcal M}_i$ and ${\mathcal M}_j$,  satisfy $i+2\le j$. Then 
\begin{align*}
L_1&=C_1\cup\ldots \cup C_i\cup P_{i+1}\cup\ldots\cup P_m,\\ 
L_2&=D_1\cup\ldots \cup D_j\cup P_{j+1}\cup\ldots\cup P_m,
\end{align*}
where $C_i\subseteq P_i$, $D_i\subseteq P_i$, moreover, by Lemma~\ref{shiftmaxl} (i) we have $|C_1|+\ldots+|C_i|=k_i-1$ and $|D_1|+\ldots+|D_j|=k_j-1$ with $|C_i|<n_i$ and $|D_j|<n_j$.
We note that Lemma~\ref{shiftmaxl} also implies that $|C_s|\le |D_s|$ for $s\in [i]$. 
%If only $|C_s|>|D_s|$ for some $s\in [i]$, and we can take $s$ the smallest with this property, then $|D_s|<n_s$ and $|D_{s+1}|=\ldots=|D_j|=0$. Since $|C_s|>0$ by (ii) we have $|C_1|+\ldots +|C_{s-1}|=n_1+\ldots+n_{s-1}$. This implies $|C_1|+\ldots+|C_s| > |D_1|+\ldots+|D_s| = k_j-1\ge k_i$, a contradiction.

%Consider the following two cases: (a) $ \emptyset \ne D_i\ne P_i$, (b) $D_i=P_i$.

%(a) Suppose $s\in [i]$ is the largest positive integer with $D_s \ne \emptyset$. Then $ |D_1|+\ldots+|D_s|= |D_1|+\ldots+|D_j|= k_j-1>k_i$. Since $|C_s|<|D_s|$ a transfer of a player  $x$ from $D_s$ to $C_s$ is possible without violating the condition $ |D_1|+\ldots+|D_t|\ge k_t$ for $t\in [j-1]$. If we could now transfer $y$ and $z$ from $L_1$ belonging to $P_j$ to $D_j$, then we would have a certificate of non-weightedness. The only condition that may prevent this is $|D_j|=n_j-1$. This would imply $k_1=n_1$, e.g., existence of veto players which we assumed not to exist.

%(b) Let $s$ be the maximal positive integer such that $|D_s|=n_s$. In this case $|D_1|+\ldots+|D_i|= n_1+\ldots+n_i \ge  k_i+1$ so we can transfer a player $x$ from $L_2$ to $L_1$ 

Let $s\in [j]$ be the largest positive integer with $D_s \ne \emptyset$. Then $ |D_1|+\ldots+|D_s|= |D_1|+\ldots+|D_j|= k_j-1$. Suppose $s\le i$.  As $k_j-1\ge k_{j-1}> k_i$, we have $|C_s|<|D_s|$.  If $s>i$, then $D_i=P_i$ and $|C_i|<|D_i|$ is also true. Thus, a transfer of a player, say  $x$, from $D_s$ to $C_s$ in the first case and from $D_i$ to $C_i$ in the second,  is possible. Let us do this transfer and, for simplicity, let us keep notation for these sets unchanged. 

After the transfer we will have 
\begin{equation}
\label{lackof2}
|D_1|+\ldots+|D_{j}|= k_{j}-2,
\end{equation}
i.e., the $j$th threshold will become further from reach but let us show that for any $t\in [j-1]$ the $t$-th threshold is still reached for $L_2\setminus \{x\}$. If $D_{j-1}=\emptyset$, this is clear. If $D_{j-1}\ne \emptyset$, then by  Lemma~\ref{shiftmaxl} (ii) we have $D_r=P_r$ for all $r\in [j-2]$. Due to \eqref{sumofns} we see that only $(j-1)$th threshold may be violated by $L_2\setminus \{x\}$. After the transfer we have, however, by \eqref{sumofns}
\[
|D_1|+\ldots+|D_{j-1}|=n_1+\ldots +n_{j-1}-1\ge k_{j-1}+(j-3)\ge k_{j-1}
\]
since $j\ge i+2\ge 3$.\par\smallskip

%After the transfer  the conditions $ |D_1|+\ldots+|D_t|\ge k_t$ for $t\in [j-2]$ would not be violated. And if $k_{j-1}>k_j-1$, then the condition  $ |D_1|+\ldots+|D_{j-1}|\ge k_{j-1}$ would not be violated either. If, however, $k_{j-1}=k_j-1$, then after the transfer we will have $ |D_1|+\ldots+|D_{j-1}|= k_{j-1}-1$. And, of course, $ |D_1|+\ldots+|D_{j}|= k_{j}-2$. We claim that in such a case $D_{j-1}\ne P_{j-1}$. Indeed, if $D_{j-1} = P_{j-1}$, then $D_\ell=P_\ell$ for all $\ell\in [j-1]$. But then by \eqref{sumofns} $ |D_1|+\ldots+|D_{j-1}|=n_1+\ldots +n_{j-1}>k_{j-1}$ which gives us a contradiction.

Now since $i<j-1$ we have $L_1\supset P_j$. If $D_{j-1}=P_{j-1}$, then by~\eqref{sumofns}  and~\eqref{lackof2}
\[
n_j-|D_j| =n_1+\ldots+n_j-(k_j-2)\ge 2,
\]
and by \eqref{n_i>1} there is a capacity to move two elements from $C_j$ to $D_j$.
If $D_{j-1}\ne P_{j-1}$, then $D_j=\emptyset$ and again there is a capacity to move two elements from 
$C_j$ to $D_j$. Suppose the elements transferred are $y,z\in C_j$.  This transfer will make $(L_2\setminus \{x\})\cup \{y,z\}$  winning. 

Let us now notice that for $L_1$ before the transfer we had
\[
|C_1|+\ldots +|C_i|+|P_{i+1}|+\ldots+|P_r|\ge k_r +(r-i)
\]
for any $r\ge j$. Indeed, by \eqref{sizeofk_i} 
\begin{align*}
&|C_1|+\ldots +|C_i|+|P_{i+1}|+\ldots+|P_r|=k_i-1 + \sum_{t=i+1}^r n_t\\ %n_{i+1}+\ldots+ n_s\\
&\ge k_i-1+ \sum_{t=i+1}^r (k_{t}-k_{t-1}-1)= k_r +(r-i).
\end{align*}
This means that $(L_1\cup \{x\})\setminus \{y,z\}$ is winning as well. 
We obtained a certificate of non-weightedness
\[
((L_1\cup \{x\})\setminus \{y,z\}, (L_2\setminus \{x\})\cup \{y,z\}   ;L_1,L_2)
\]
which gives us a contradiction. 
\end{proof}

The only way we can grow the dimension in hierarchical conjunctive games is to increase the number $m$ of equivalence classes. This was not the case whith disjunctive hierarchical games. 

Theorems~\ref{OS2} and~\ref{OS} demonstrate that, inter alia,  the dimension is not preserved under duality. Indeed, it is known \citep{gvozdeva13} that the duality takes us from disjunctive hierarchical games to  conjunctive ones and vice versa and their respective growth of dimension are very different. \par\medskip

{\bf Codimension.}\par\medskip

A concept closely related to dimension is that of the codimension. 

\begin{definition}% \cite{freixas09}
The \emph{codimension} of a simple game is the minimum number of weighted simple games whose union forms the given game. That is, the simple game $G = (P, W)$ has codimension $n$ if $W = W_1 \cup \cdots \cup W_n$, where each of the games $(P, W_1), \ldots, (P, W_n)$ is weighted, and $W$ cannot be represented as the union of fewer than $n$ weighted games. We will denote the codimension of $G$ by $\text{codim}(G)$.
\end{definition}

This concept emerges in relation to the duality of games.

\begin{theorem}[Freixas-Marciniak, 2009]
If $G $ is a simple game, then $$\text{codim}(G^\ast) = \dim(G).$$
\end{theorem}

\begin{proof}
This is a simplification of Theorem 3.2(ii) in \cite{freixas09}, in which we take $\mathcal{C}$ to be the class of weighted simple games, along with their observation that this class is closed under duality.
\end{proof}

Due to Freixas-Marciniak theorem we can extract some consequences from our results with respect to codimension.

\begin{corollary}
The codimension $d$ of an $m$-partite disjunctive  hierarchical game  satisfies  ${\lceil \frac{m}{2} \rceil -1 \le d \le m}$. 
\end{corollary}

\begin{proof}
Follows from Theorem~\ref{OS} and the fact that the dual game to a disjunctive  hierarchical game is conjunctive hierarchical \citep{gvozdeva13}.
\end{proof}

\begin{corollary}
There is a sequence of conjunctive hierarchical games whose codimensions grow exponentially in the number of players. 
\end{corollary}

\section{Problems with the concept of dimension}

There are several worrying properties of the concept of dimension. One of these can be observed from the results of this paper, namely, that the dimension is not preserved under duality. Indeed, in \cite{gvozdeva13} it was proved that the dual of a  disjunctive hierarchical games is a conjunctive hierarchical games and vice versa.  However, Theorems~\ref{OS2} and~\ref{OS} show that the dimensions of games in these classes are very different. Another problem with this concept is illustrated in the following example. \par\smallskip

{\bf Motivating example.}\footnote{The example we present in this section is the result of an email discussion with Bill Zwicker.} It shows that when one represents a complete game $G$ as an intersection of weighted games, it may be impossible to choose the weightings in a way that faithfully represents the desirability order $\succeq_G$ in $G$. 

Let  $P=P_1\cup P_2$ with $|P_1|=2$ and $|P_2|=5$. Consider a disjunctive hierarchical game $H=H_\exists(P, {\bf k})$ with ${\bf k}=(2,5)$. Frstly, let us consider $P$ as a multiset $\{1^2,2^5\}$, i.e., consisting of two  identical players of type 1 and five of type 2. We ask if it is possible to find games $\row Hk$, where in the $i$th game %has voting representation $[q_i;w_i(1),w_i(2)]$, i.e., assigning 
weight $w_i(1)$ is assigned to all players of the first type, weight $w_i(2)]$ to players of the second type and the threshold is $q_i$.

Considering classes of equivalent coalitions of $H$ as submultisets,  we can graphically represent them as points in $\mathbb{R}$ representing the coalition $\{1^x,2^y\}$ as point $(x,y)$. Thus, 
 the minimal winning ones are $\{1^2\}$, $\{1,2^4\}$, $\{2^5\}$ will be represented as points $(2,0)$, $(1,4)$, $(0,5)$. We can graphically depict them on the following diagram where the area of winning coalitions is presented in grey.\par\smallskip

\begin{center} 
\includegraphics[width=4.5cm]{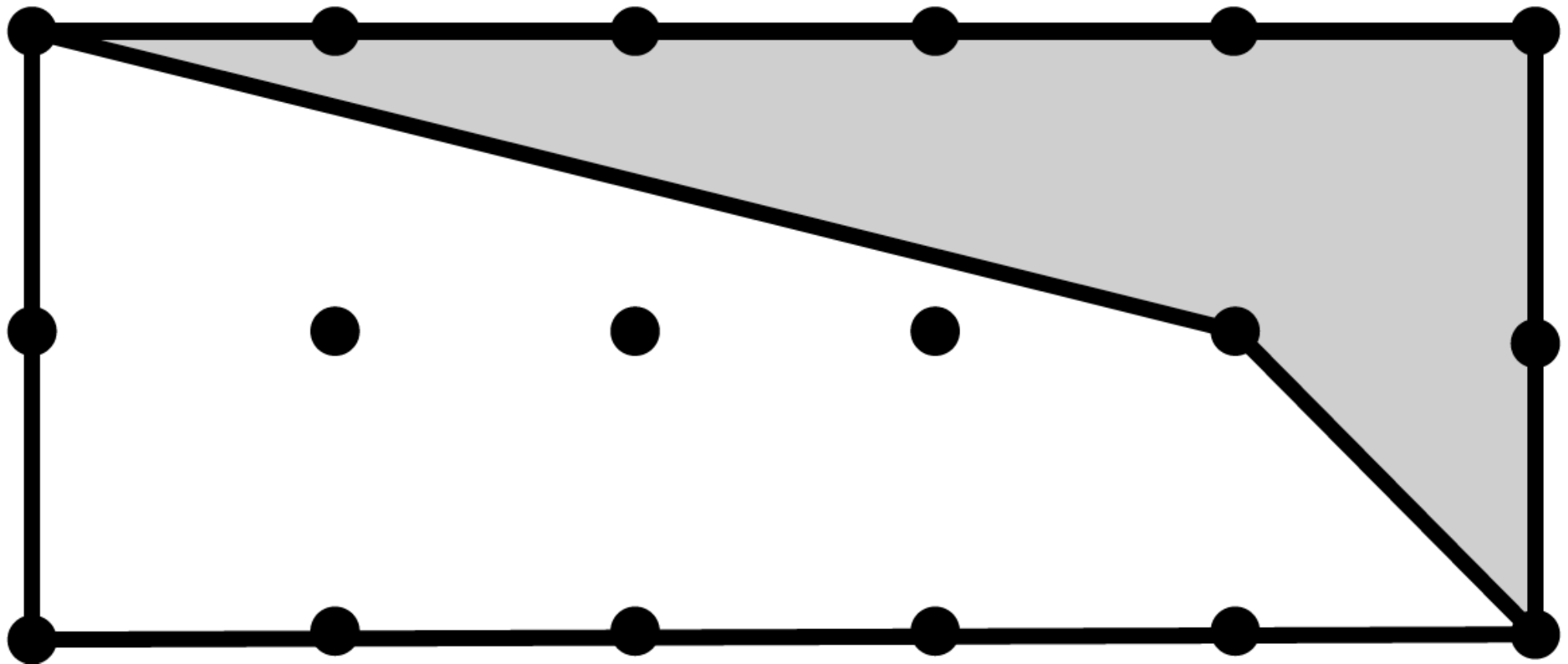}
\end{center}
\begin{center}
{\em Figure 1: Vertical axis $x$ shows the number of players of the first type and horizontal axis $y$ of the second.}
%\vspace{3mm}
\end{center}
If it was possible to find weight functions $w_1, w_2$ and thresholds $q_1,q_2$, then The set of winning coalitions of $H$ will be convex.
This grey area is not convex; its convex hull also contains the point  corresponding to the only maximal losing coalition $\{1,2^3\}$. This means that the games $\row Hk$ cannot be found. What can be done?

Let $P_1=\{b_1,b_2\}$ and $P_2=\{c_1, c_2,c_3,c_4,c_5\}$. Consider the two following representations:\par\smallskip
{\bf First representation.}
%FIRST REPRESENTATION:
%
We define two weighted games $G_1=(P,W_1)$ and $G_2=(P,W_2)$ as follows:
\begin{align*}
w_1(b_1) &= 4,\  w_1(b_2) = 1.1,\\
\forall_{j\in [5]} \ w_1(c_j) &= 1\ \text{and} \  \text{quota $= 5$}.
\end{align*}
and
\begin{align*}
w_1(b_1) &= 1.1,\  w_1(b_2) = 4,\\
\forall_{j\in [5]} \ w_1(c_j) &= 1\ \text{and} \  \text{quota $= 5$}.
\end{align*}
Obviously, $H=G_1\cap G_2$ and, since $H$ is not weighted, we have $\dim H=2$.\par\smallskip
{\bf Second representation.}
%
%SECOND REPRESENTATION:
%
The above representation failed to represent  the `equivalence' part of $\sim_H$ for the players of $P_1$, but did succeed with the $P_2$ players (in that they did get equal weight for each of the two weightings).

  Next, we look at a representation that switches roles: it similarly represents the strict part of $\sim_H$  for all players, and it gives the $P_1$ players equal weight for each of the weightings (but the five players of $P_2$ get different weights). For every subset $X\subset [5]$ such that $|X|=3$ we define a game $G_X=(P,W_X)$ by
\begin{align*}
&w_X(b_1) = w_X(b_2) = 3,\\
&\forall_{i\in X}\ w_X(c_i) = 2,\  \forall_{i\notin X}\ w_X(c_i) = 0,\\
&quota = 6
\end{align*}
There are 10 such sets $X$ of cardinality 3, and thus 10 weighted games, and it is easy to show that their intersection is the hierarchical game $H$.  The nice property of this particular vector of weightings is that it respects the `strict' part $\sim_H$ of the individual desirability order of $H$ in the following sense: $x \prec_H y$ iff $w_X(x) < w_X(y)$ holds for each three-element subset $X$ of $[5]$.
%Of course, for the $P_1$ players it does not respect the weak part of $\succeq_H$.  

This example explains why the lower bound in the theorem of Freixas and Puente (Theorem~2) had to be reproved. In their proof of the lower bound they allowed only weighted games that assign equal weights to players who are equivalent in the original game. We see that this was not sufficient to claim that the lower bound in their theorem holds.\par\smallskip

%\begin{problem}
%Let $G=(P,W)$ be a complete game with equivalence classes for the desirability relation $\row Pm$ so that $P=P_1\cup\ldots, \cup P_m$. Then for eash $i\in [m]$ there exists a representation of $G$ as an intersection of weighted games such that each game in this representation assigns equal weights to players from $P_i$ and respects the strict  part of the desirability order $\succeq_G$?
%\end{problem}

{\bf Boolean dimension of simple games.} \par\medskip

Boolean dimension of a simple game was introduced in \cite{faliszewski09}. Let $ \Phi=\{p,q,\ldots \}$ be a set of propositional variables and let $\mathcal{L}$ denote the set of (well-formed)
formulas of the first-order propositional logic over $\Phi$ containing only logical connectives $\wedge$ and $\vee$\footnote{In paper by \cite{faliszewski09}  negations were also allowed, however, the authors of that paper considered also non-monotonic simple games which we do not consider.}. For a formula $\phi\in \mathcal{L}$ let $|\phi |$ be the number of variables used to express $\phi $. Suppose also $\top$ is a tautology and $\bot$ is a contradiction. Let $G_i=(P,W_i)$, $i=1,\ldots,q$, be simple games with the same set of players $P$. We will define the game $G=(P,W)$ by setting for a coalition $C\subseteq P$
\[
C\in W := \phi (C\in W_1,\ldots, C\in W_q) =\top.
\]
We will denote this game $\phi (G_1,\ldots,G_q)$.
We illustrate this definition with the following simple games which plays an important role in the theory of secret sharing \citep{beimel:360}. They are called there tripartite.

\begin{example}%[Tripartite games in secret sharing]
\label{extrip}
Let ${\bf n}=(n_1,n_2,n_3)$ and ${\bf k}=(k_1,k_2,k_3)$, where $n_1,n_2,n_3$ and $k_1,k_2,k_3$ are positive integers. The game $\Delta_1({\bf n},{\bf k})$ is defined on the set $P=P_1\cup P_2\cup P_3$ which is a union of disjoint sets $P_1,P_2,P_3$ of cardinalities $n_1,n_2,n_3$, respectively,. A coalition $C=C_1\cup C_2\cup C_3$, where $C_i\subseteq P_i$, $i=1,2,3$, is winning iff
\[
 (|C_1|\ge k_1) \vee [(|C_1|+|C_2|\ge k_2)\wedge (|C_1|+|C_2|+|C_3|\ge k_3)] = \top,
\]
where $k_1<k_3,\quad k_2<k_3,\quad n_1 \geq k_1,\quad n_2 >k_2- k_1$ and $ n_3> k_3-k_2.$ Obviously, it is organised as $G_1\vee (G_2\wedge G_3)$, where $G_1,G_2,G_3$ are weighted games. 
\end{example}

\begin{definition}
Let $G$ be a simple game. The smallest positive integer $d$ such that $G$ can be represented as $G=\phi (G_1,\ldots,G_n)$, where $\row Gn$ are weighted simple games, and $|\phi |=d$ is called the Boolean dimension of $G$.
\end{definition}

The Boolean dimension of the game $\Delta_1({\bf n},{\bf k})$ from Example~\ref{extrip} is obviously 3 while its classical dimension may (and certainly will) depend on the parameters ${\bf n}$ and ${\bf k}$. The revised voting rules of the Council of the European Union also has Boolean dimension~3.

Apart from a better ability to reflect the descriptive complexity of games, the Boolean dimension has some nice properties absent in the classical dimension.

\begin{proposition}
The Boolean dimension of a simple game is equal to the Boolean dimension of its dual.
\end{proposition}

\begin{proof}
Let $G_i=(P,W_i)$, $i=1,2$, be a simple game and $G_i^\ast =(P,W_i^\ast)$, $i=1,2$, be their dual games. \cite{taylor99} (see, e.g., Proposition~1.4.3) showed that the following two de Morgan laws are satisfied
\begin{align*}
(W_1\cup W_2)^\ast &= W_1^\ast \cap W_2^\ast,\\
(W_1\cap W_2)^\ast &= W_1^\ast \cup W_2^\ast.
\end{align*}
From here it immediately follows that 
\begin{align*}
(G_1\vee G_2)^\ast &= G_1^\ast \wedge G_2^\ast,\\
(G_1\wedge G_2)^\ast &= G_1^\ast \vee G_2^\ast,
\end{align*}
which imply the statement.
\end{proof}

However, some features of Boolean dimension are the same as the corresponding features of dimension. In particular, as is the case for dimension, in general the calculation of the exact Boolean dimension of a simple game is NP-hard~\citep{faliszewski09}.

\section{\!\!\!\! Conclusion and open questions}
\smallskip

We have answered a question of \cite{freixas08} by showing that the dimension of a complete simple game can grow exponentially in the number of players. The games used to demonstrate this are the disjunctive hierarchical games. We have discussed the pitfalls of the concept of dimension and found some disturbing features of it which led us to reproving and strengthening the theorem of Freixas and Puente about the range in which the dimension of a conjunctive hierarchical game may lie. This research prompts the following questions:
\begin{itemize}
\item Can anything be said about the Boolean dimension of a complete game?
\item Let $G=(P,W)$ be a complete game with equivalence classes for the desirability relation $\row Pm$ so that $P=P_1\cup\ldots, \cup P_m$. Then for each $i\in [m]$ does there exist a representation of $G$ as an intersection of weighted games such that each game in this representation assigns equal weights to players from $P_i$ and respects the strict  part of the desirability order $\succeq_G$?
\end{itemize}

\section{Acknowledgements}

Arkadii Slinko was  supported by the Royal Society of NZ  Marsden Fund UOA-254.  He also thanks Bill Zwicker for fruitful discussions of the concept of dimension.

%\section{{\color{black}The Bibliography}}
\bibliographystyle{elsarticle-harv}
\bibliography{dimension}

%\addcontentsline{toc}{chapter}{Appendix}
%\begin{appendix}

\end{document}